\documentclass[11pt]{article}
\usepackage{amssymb}
\usepackage[perpage,symbol]{footmisc}
\usepackage{listings}
\usepackage{amsfonts}
\usepackage{enumerate}
\usepackage{amsmath}
\usepackage{cite}
\usepackage{array}
\usepackage{booktabs}

\begin{document}

\newtheorem{theorem}{Theorem}[section]
\newtheorem{corollary}[theorem]{Corollary}
\newtheorem{definition}[theorem]{Definition}
\newtheorem{proposition}[theorem]{Proposition}
\newtheorem{lemma}[theorem]{Lemma}
\newtheorem{example}[theorem]{Example}
\newtheorem{conjecture}[theorem]{Conjecture}
\newenvironment{proof}{\noindent {\bf Proof.}}{\rule{3mm}{3mm}\par\medskip}
\newcommand{\remark}{\medskip\par\noindent {\bf Remark.~~}}
\title{Higher weight distribution of Linearized Reed-Solomon codes}
\author{Haode Yan\footnote{Dept. of Math., Shanghai Jiao Tong Univ., Shanghai 200240, hdyan@sjtu.edu.cn.},\ Yan Liu\footnote{Dept. of Math., SJTU, Shanghai 200240, liuyan0916@sjtu.edu.cn.},\  Chunlei Liu\footnote{Corresponding author, Dept. of Math., SJTU, Shanghai 200240, 714232747@qq.com.}}
\date{}
\maketitle
\thispagestyle{empty}

\abstract{Let $m$, $d$ and $k$ be positive integers such that $k\leq \frac{m}{e}$, where $e=(m,d)$. Let $p$ be an prime number and $\pi$  a primitive element of ${\mathbb F}_{p^m}$. To each $\vec {a}=(a_{0}, \cdots, a_{k-1})\in \mathbb{F}_{p^{m}}^{k}$, we associate the linearized polynomial
$$f_{\vec {a}}(x)=\sum_{j=0}^{k-1}a_{j}x^{p^{jd}},$$
as well as the sequence
\[c_{\vec {a}}=(f_{\vec {a}}(1),f_{\vec {a}}(\pi), \cdots , f_{\vec {a}}(\pi^{p^m-2})). \]
Let \[C=\{c_{\vec a}\mid~\vec{a}\in{\mathbb F}_{p^m}^k\}\] be the cyclic code formed by the sequences $c_{\vec {a}}$'s. We call $C$ a linearized Reed-Solomon code. The higher weight distribution of the code $C$
is determined in the present paper.



\section{\small{INTRODUCTION}}
\paragraph{}
Let $q$ be a prime power, and $C$ an $[n,k]$-linear code over ${\mathbb F}_q$. For an ${\mathbb F}_q$-subspace $H$ of $C$, the weight of $H$ is defined to be
\[{\rm wt}(H)=\#\{0\leq i\leq n-1|~c_i\neq0,~\text{ for some }(c_0,c_1,\cdots,c_{n-1})\in H\}.\]
If $\dim H=r$, ${\rm wt}(H)$ is called an $r$-dimensional weight of $C$. If $r>1$, an $r$-dimensional weight of $C$ is called a higher weight of $C$.
\paragraph{}
The notion of higher weights of a linear code was first introduced by Helleseth-Klove-Mykkeltveit \cite{HKM} in 1977. Motivated by Ozarow-Wyner's paper \cite{OW}  on the cryptographical significance of a linear code in a wire-tap channel, Wei \cite{Wei} rediscovered the notion of higher weights of a linear code in 1991. Since then, the study of the higher weights of linear codes has attracted lots of attention, see, for example, the papers \cite{A-B-L14},\cite{B-L-V14},\cite{C-C97},\cite{C-M-T01},\cite{C91},\cite{C-L-Z94},\cite{D-F-G14},\cite{D-S-V96},
\cite{F-T-W92},\cite{GO},
\cite{H-K92},\cite{H-K09},\cite{H-K-L95},\cite{J-L97},\cite{J-L07},\cite{M-P-P98},\cite{M-R99},\cite{M94},\cite{S-C95},
\cite{S-W03},\cite{S-V94},\cite{T-V95},\cite{G-V94},\cite{G-V942},\cite{G-V95},\cite{G-V952},\cite{G-V953},\cite{G-V96}, \cite{V96} and \cite{YKS}.
\paragraph{}
For $1\leq r\leq k$, and for $1\leq w\leq n$, the number of $r$-dimensional subspaces of $C$ of weight $i$ is
\[n_{r,w}=\#\{H\mid {\rm wt}(H)=w,~\dim H=r\}.\]
Given a $[n,k]$-linear code $C$, it is challenging to determine the set
\[\{n_{r,w}\mid~1\leq r\leq k,~ n-k+r\leq w\leq k\},\]
which is called the higher weight distribution of $C$.
The higher weight distribution is known only for a few classes of codes. Helleseth-Kl\o ve-Mykkeltveit \cite{HKM} determined the higher weight distribution of the MDS code. Kl\o ve \cite{Kl} determined the higher weight distribution of the binary $[23,11]$-Golay code. Helleseth \cite{He} and Hirschfeld-Tsfasman-Vladut \cite{HTV} determined the higher weight distribution of some other classes of codes.
\paragraph{}
For $1\leq r\leq k$, the minimum $r$-dimensional weight of $C$ is
\[d_r=\min\{{\rm wt}(H)\mid~\dim H=r\}.\]
Given a $[n,k]$-linear code $C$, it is significant to determine the set
$\{d_1,d_2,\cdots,d_k\}$,
which is called the weight hierarchy of $C$.
The weight hierarchy, though an easier problem than the higher weight distribution, is still known only for a few classes of codes.
Helleseth-Kumar \cite{HK95} determined the weight hierarchy of the Kasami code.
Helleseth-Kumar \cite{HK96},
Vlugt \cite{Vl95} and Yang-Li-Feng-Lin \cite{YLFL} determined the weight hierarchy of irreducible cyclic codes.
Heijnen-Pellikaan \cite{HP} determined the weight hierarchy of the Reed-Muller code. Barbero-Munuera \cite{BM} determined the weight hierarchy of Hermitian codes. Xiong-Li-Ge determined the weight hierarchy of some reducible cyclic codes. Wei-Yang \cite{WY}, Helleseth-Kl\o ve \cite{HK}, Park \cite{Par} and Martinez-Perez-Willems \cite{MW} determined the weight hierarchy of some product codes.
\paragraph{}
In the present paper we shall define the linearized Reed-Solomon code and determine its higher weight distribution.
Let $m$, $d$ and $k$ be positive integers such that $k\leq \frac{m}{e}$, where $e=(m,d)$. Let $p$ be an prime number and $\pi$  a primitive element of ${\mathbb F}_{p^m}$. To each $\vec {a}=(a_{0}, \cdots, a_{k-1})\in \mathbb{F}_{p^{m}}^{k}$, we associate the linearized polynomial
$$f_{\vec {a}}(x)=\sum_{j=0}^{k-1}a_{j}x^{p^{jd}},$$
as well as the sequence
\[c_{\vec {a}}=(f_{\vec {a}}(1),f_{\vec {a}}(\pi), \cdots , f_{\vec {a}}(\pi^{p^m-2})). \]
From now on we write \[C=\{c_{\vec a}\mid~\vec{a}\in{\mathbb F}_{p^m}^k\}.\]  We call it a linearized Reed-Solomon code.
Our preliminary result is  the following.
\begin{theorem}\label{valueset} If $H$ is a ${\mathbb F}_{p^m}$-subspace of $C$ of dimension $r>0$, then
\[{\rm wt}(H) \in \{p^m-p^{ei}|\  0\leq i\leq k-r \}.\]
\end{theorem}
Our main result is the following.
\begin{theorem}\label{frequency}If $r>0$, and $0\leq i\leq k-r$, then
\[n_{r,p^m-p^{ei}}=\binom{\frac{m}{e}}{i}_{p^e}\sum_{j=0}^{k-r-i} (-1)^{j}p^{e\binom{j}{2}}\binom{k-j-i}{r}_{p^m}\binom{\frac{m}{e}-i}{j}_{p^e},\]
where $\binom{n}{i}_{q}$ is the number of $i$-dimensional $\mathbb{F}_{q}$-subspaces of $\mathbb{F}_{q}^{n}$.
In particular, $d_r=p^m-p^{e(k-r)}$.
\end{theorem}
The classical weight distribution of the linearized Reed-Solomon code follows from a result of
Delsarte \cite{Del}  on the rank distribution of bilinear forms.

\section{\small{LINEARIZED VAN DER MONDE MATRIX}}
In this section, we will introduce the notion of linearized Van Der Monde matrices.
\begin{lemma}\label{nullset}If $\vec {a}\in{\mathbb F}_{p^m}^k$ is nonzero, then the number of zeros of
$f_{\vec {a}}(x)$ is $\leq p^{e(k-1)}$.
\end{lemma}
\begin{proof} Suppose that $\vec {a} \neq 0$. Note that $\{x \in  \mathbb{F}_{p^{md/e}}|f_{\vec {a}}(x)=0\}$ is a subspace of $\mathbb{F}_{p^{md/e}}$ over ${\mathbb F}_{p^d}$ of dimension $\leq (k-1)$.
As $(m, d)=e$,
a basis of $\mathbb{F}_{p^{m}}$ over ${\mathbb{F}_{p^{e}}}$ is also a basis of $\mathbb{F}_{p^{md/e}}$ over $\mathbb{F}_{p^{d}}$.
It follows that, the ${\mathbb F}_{p^e}$-space consisting of the zeros of $f_{\vec a}(x)$ is of dimension $\leq k-1$. The lemma now follows.
\end{proof}

\begin{definition}\label{van}If $(m,d)=e$, and $x_{1}, x_{2},\cdots,x_{k}$ are $\mathbb{F}_{p^{e}}$-linearly independent elements in  $\mathbb{F}_{p^{m}}$, then matrix\[
\left( \begin{array}{cccc}
x_1&{x_1}^{p^d}&\cdots&{x_1}^{p^{(k-1)d}}\\
x_2&{x_2}^{p^d}&\cdots&{x_2}^{p^{(k-1)d}}\\
\vdots&\vdots&\ddots&\vdots\\
x_k&{x_k}^{p^d}&\cdots&{x_k}^{p^{(k-1)d}}\\
\end{array} \right)\] is called $p^e$-linearized Van Der Monde matrix.
\end{definition}
The following lemma is a slight generalization of a lemma of Cao-Lu-Wan-Wang-Wang \cite{CLW}.
\begin{lemma}\label{rank} A linearized Van Der Monde matrix is of full rank.
\end{lemma}
\begin{proof} Suppose that $(m,d)=e$, $x_{1}, x_{2},\cdots,x_{k}$ are $\mathbb{F}_{p^{e}}$-linearly independent elements in  $\mathbb{F}_{p^{m}}$, and
  $a_0,a_1,\cdots,a_{k-1}$ are elements of ${\mathbb F}_{p^m}$ such that
\[a_{0}\left(
\begin{array}{c}
x_{1}\\
x_{2}\\
\cdots\\
x_{k}\\
    \end{array}
  \right)+a_{1}\left(
\begin{array}{c}
x_{1}^{p^{d}}\\
x_{2}^{p^{d}}\\
\cdots\\
x_{k}^{p^{d}}\\
    \end{array}
  \right)+\cdots+a_{k-1}\left(
\begin{array}{c}
x_{1}^{p^{(k-1)d}}\\
x_{2}^{p^{(k-1)d}}\\
\cdots\\
x_{k}^{p^{(k-1)d}}\\
    \end{array}
  \right)=0.\]
  Then $x_1,x_2,\cdots,x_{k}$ are zeros of $f_{\vec a}$. It follows that $|{\rm Null}(f_{\vec {a}})|\geq p^{ek}$.  By Lemma \ref{nullset},
  we have $a_0=a_1=\cdots=a_{k-1}=0$.
  Therefore \[\left(
\begin{array}{c}
x_{1}\\
x_{2}\\
\cdots\\
x_{k}\\
    \end{array}
  \right),\left(
\begin{array}{c}
x_{1}^{p^{d}}\\
x_{2}^{p^{d}}\\
\cdots\\
x_{k}^{p^{d}}\\
    \end{array}
  \right),\cdots,\left(
\begin{array}{c}
x_{1}^{p^{(k-1)d}}\\
x_{2}^{p^{(k-1)d}}\\
\cdots\\
x_{k}^{p^{(k-1)d}}\\
    \end{array}
  \right)\]
  are linearly independent over ${\mathbb F}_{p^m}$.
  It follows that the matrix\[
\left( \begin{array}{cccc}
x_1&{x_1}^{p^d}&\cdots&{x_1}^{p^{(k-1)d}}\\
x_2&{x_2}^{p^d}&\cdots&{x_2}^{p^{(k-1)d}}\\
\vdots&\vdots&\ddots&\vdots\\
x_k&{x_k}^{p^d}&\cdots&{x_k}^{p^{(k-1)d}}\\
\end{array} \right)\] is of full rank. The lemma is proved.
\end{proof}
The above lemma implies the following.
\begin{corollary}\label{rankcor}If $i\leq k$ and $x_{1},x_{2},\cdots,x_{i}$ are $\mathbb{F}_{p^{e}}$-linearly independent elements in  $\mathbb{F}_{p^{m}}$, then the matrix \[
\left( \begin{array}{cccc}
x_1&{x_1}^{p^d}&\cdots&{x_1}^{p^{(k-1)d}}\\
x_2&{x_2}^{p^d}&\cdots&{x_2}^{p^{(k-1)d}}\\
\vdots&\vdots&\ddots&\vdots\\
x_i&{x_i}^{p^d}&\cdots&{x_i}^{p^{(k-1)d}}\\
\end{array} \right)\] is of full rank over ${\mathbb F}_{p^m}$.
\end{corollary}
The above corollary implies the following.
\begin{corollary}\label{dimbound}If $H$ is an ${\mathbb F}_{p^m}$-subspace of ${\mathbb F}_{p^m}^k$ of dimension $r>0$, then
\[Z(H)=\{ x \in \mathbb{F}_{p^{m}}|~ f_{\vec {a}}(x)=0,~\forall \vec{a}\in H\}\]
an ${\mathbb F}_{p^e}$-subspace of ${\mathbb F}_{p^m}$ of dimension $\leq k-r$.
\end{corollary}
The above corollary implies the following.
\begin{corollary}\label{precounting} If $r>0$, and $U$ is an ${\mathbb F}_{p^e}$-subspace of ${\mathbb F}_{p^m}$ of dimension $i$, then
\[
|C_{r,U}|=
\begin{cases}
0, &  i >k-r,\\
\binom{k-i}{r}_{p^m}, & i \leq k-r,
\end{cases}
\]
where
\[C_{r,U}=\{ H\subseteq \mathbb{F}_{p^{m}}^{k}|~Z(H)\supseteq U,~\dim_{{\mathbb F}_{p^m}}H=r\}.\]
\end{corollary}
\section{\small{PROOF OF THE THEOREMS}}
\paragraph{}
It is easy to see that
\[{\rm wt}(H)=p^m-p^{e\dim_{{\mathbb F}_{p^e}} Z(H)}.\]
Theorem \ref{valueset} now follows from Corollary \ref{dimbound}.
\paragraph{}
We now prove Theorem \ref{frequency}.
For an ${\mathbb F}_{p^e}$-subspace $U$ of ${\mathbb F}_{p^m}$, we write
\[S_{r,U}=\{ H\subseteq \mathbb{F}_{p^{m}}^{k} |~Z(H)=U,~\dim_{{\mathbb F}_{p^m}}H=r\}.\]
Then
\[n_{r,p^m-p^{ei}}=\sum_{\dim_{{\mathbb F}_{p^e}}U=i}|S_{r,U}|.\]
By definition,
\[|C_{r,W}|=\sum_{W\subseteq U}|S_{r,U}|.\]
By the $q$-binomial M\"{o}bius inversion formula,
\[|S_{r,W}|=\sum_{W \subseteq U}(-1)^{\rm{dim} U/W}p^{e\binom{\rm{dim} U/W}{2}}|C_{r,U}|.\]
It follows that
\[n_{r,p^m-p^{ei}}=\binom{\frac{m}{e}}{i}_{p^e}\sum_{j=0}^{k-r-i} (-1)^{j}p^{e\binom{j}{2}}\binom{k-j-i}{r}_{p^m}\binom{\frac{m}{e}-i}{j}_{p^e}.\]
Theorem \ref{frequency} is proved.

\end{document}